\documentclass[a4paper]{article}
\usepackage[utf8]{inputenc}
\usepackage{amsmath}
\usepackage{amsfonts}
\usepackage{amsthm}
\usepackage{listings}
\usepackage{mdframed}
\usepackage{url}

\newtheorem{theorem}{Theorem}
\newtheorem{lemma}{Lemma}
\theoremstyle{definition}

\newtheorem{definition}{Definition}

\newcommand{\s}{\mathbf{scan}}
\renewcommand{\ss}{\mathbf{sscan}}
\newcommand{\zip}{\mathbf{zip}}
\newcommand{\sort}{\mathbf{sort}}
\renewcommand{\ss}{\mathbf{sscan}}
\newcommand{\map}{\mathbf{map}}
\newcommand{\flatmap}{\mathbf{flatmap}}
\renewcommand{\zip}{\mathbf{zip}}
\renewcommand{\sort}{\mathbf{sort}}
\newcommand{\shift}{\mathbf{shift}}
\newcommand{\mzip}{\mathbf{mapzip}}
\newcommand{\fmzip}{\mathbf{flatmapzip}}
\newcommand{\bm}{\mathbf{broadcastmax}}

\newcommand{\RR}{\mathbb{R}}
\DeclareMathOperator*{\F}{F}

\title{Aggregating over Dominated Points by Sorting, Scanning, Zip and Flat Maps}
    \author{Jacek Sroka \and Jerzy Tyszkiewicz}
    \date{University of Warsaw}
\begin{document}
\begin{titlepage}

\maketitle
\begin{abstract}
    Prefix aggregation operation (also called scan), and its particular case, prefix summation, is an important parallel primitive and enjoys a lot of attention in the research literature. It is also used in many algorithms as one of the steps. 
    
    Aggregation over dominated points in $\RR^m$ is a multidimensional generalisation of prefix aggregation. It is also intensively researched, both as a parallel primitive and as a practical problem, encountered in computational geometry, spatial databases and data warehouses.
    
    In this paper we show that, for a constant dimension $m$, aggregation over dominated points in $\RR^m$ can be computed by $O(1)$ basic operations that include sorting the whole dataset, zipping sorted lists of elements, 
    computing prefix aggregations of lists of elements and flat maps, which expand the data size from initial $n$ to $n\log^{m-1}n$.
    
    Thereby we establish that prefix aggregation suffices to express aggregation over dominated points in more dimensions, even though the latter is a far-reaching generalisation of the former. Many problems known to be expressible by aggregation over dominated points become expressible by prefix aggregation, too. 
    
    We rely on a small set of primitive operations which guarantee an easy transfer to various distributed architectures and some desired properties of the implementation.
\end{abstract}
\end{titlepage}

\section{Introduction}

In this paper we derive a tight relation between two computing problems: prefix aggregation and aggregation over dominated points. We first describe the problems alone, in a framework which covers them both. 

The input data is assumed to be a collection of tuples composed of sortable atomic elements. We are interested in aggregation in general.  We assume the data consists of two sets: a set $D$ of data points and a set $Q$ of queries where for each query there will be a subset $\hat q\subseteq D$ of data points it matches.% with successive atomic elements in a predefined relation to the atomic elements of the query, e.g., smaller or larger or equal. 

We are also given an associative and commutative function $\F: A \times A \to A$ with unit $e$, i.e., neutral element of $\F$. We use $\F_{d\in S} w(d)$ and the like for the application of $\F$ to a set of elements $S \subseteq D$, the same way as $\Sigma$ is used as a generalisation of $+$ to a multiset of numbers.

Data points have weights in $A$, defined by a function $w:D \to A$. The goal is to compute $\F_{d\in \hat q}w(d)$ for all $q\in Q.$ The two problems we consider in the paper are specific cases of the above schematic outline. We assume $D$ and $Q$ to be large.

\subsection{Prefix aggregation and its role as a parallel primitive}

Prefix aggregation (also called scan) arises when data and queries are linearly ordered, say are both elements of $\RR$, and $\hat q=\{d\in D~|~d< q\}$. Then we wish to compute $\F_{d<q}w(d).$ The eponymous case is when $D=Q$, so the aggregations are applied to all prefixes of the whole sequence of data points.

Concerning the importance of prefix aggregation as a computing primitive, it is well-known that, no matter what practical or theoretical parallel computation model is considered, sorting 
and scan are among the very first algorithms to be developed. The importance of the latter has even led to patents around the idea of including a prefix sum in the instruction set of a microprocessor~\cite{vishkin-patent}, attempts of hardware implementations~\cite{lin,scan-hardware1,scan-hardware2} or adapting the actual algorithm to the architecture of the processor, sequential~\cite{lai2022efficient} or parallel~\cite{scan-CUDA}. Attempts of formal verification~\cite{scan-formal,hinze} also exist.

Scans were also researched as a primitive to implement parallel variants of many algorithms, including radix sort, quicksort, lexical analysis, polynomial evaluation, stream compaction, histograms and string comparison~\cite{blelloch,blelloch2}, and also geometric partitioning algorithms~\cite{Teng}.

\subsection{Aggregation over dominated points}

In this problem, data and queries are points in the $m$-dimensional space  $\RR^m$. For two such points we write $(x_1,\ldots,x_m)<(y_1,\ldots,y_m)$ when inequalities hold coordinate-wise, i.e., $x_1<y_1,\ldots,x_m<y_m.$ Then let $\hat{q}=\{d\in D|d<q\}$.

The result of the algorithm should therefore consist of $\F_{d<q} w(d)$ for each query point $q$, which is the result of applying $\F$ to the set of all weights of points in $D$ which are \emph{dominated} by $q$, i.e., coordinate-wise smaller than $q$.

Aggregation over dominated points can obviously be considered as a multidimensional generalisation of 1-dimensional prefix aggregation. Note however that typically prefix aggregation uses an associative operation $\F$ with unit, while in the multidimensional setting we also require it to be commutative. The reason is that prefix aggregation has a clear order in which the elements are aggregated. In multidimensional setting there is no such natural order and hence, for the sake of producing a deterministic result, we require commutativity.

Aggregation over dominated points, referred to as \emph{general prefix computations}, has been shown to be a parallel primitive which allows expressing many computational problems~\cite{springsteel}. This approach has been subsequently extended to a general parallel computation model called {\em Broadcast with Selective Reduction} PRAM (BSR for short). The multiple criteria variant of BSR, introduced by Akl and Stojmenović~\cite{akl&st1,akl&st2} after a number of earlier papers about single criterion BSR, is the one whose only parallel primitive is aggregation over dominated points. Many computational problems have been then shown to have constant-round BSR algorithms including: counting intersections of isothetic line segments, vertical segment visibility, maximal elements in $m$ dimensions, ECDF searching, 2-set dominance counting and rectangle containment in $m$ dimensions, rectangle enclosure and intersection counting in $m$ dimensions~\cite{akl&st1}, all nearest smaller values~\cite{xiang1998ansv}, all nearest neighbours and furthest pairs of points in a plane in $L_1$ metric, the all nearest foreign neighbours in $L_1$ metric and the all furthest foreign pairs of points in the plane in $L_1$ metric~\cite{melter1995constant}. All of them therefore can be expressed by aggregation over dominated points.

Other applications of this primitive include calculating Empirical Cumulative Distribution Functions (ECDFs) in statistics, which are required in multivariate Kolmogorov-Smirnov, Cramér-von Mises and Anderson-Darling statistical tests~\cite{ECDF2}. The problem is also intensively studied, under the name range-aggregate queries by the spatial database community~\cite{AGARWAL,tao2004range,shi2021} and data warehouse community~\cite{hong}. However, in this area one typically does not assume queries to be known in advance, so the focus is rather on storing data points in a data structure which allows efficient querying.

\subsection{Our contribution}

We prove the following.

\begin{theorem}\label{thm1}
Aggregation over dominated points in $\RR^m$, where $m$ is constant, can be computed in $O(1)$ basic operations: sorting of lists of tuples, zipping and computing prefix aggregations as well as flat maps over such lists. By using flat lists, our algorithm expands the initial size of the input data from $n$ to  $O(n\log^{m-1}n)$ tuples. 
\end{theorem}

Our work thus creates a direct link between so far separate parallel primitives. In this respect, we follow the paradigm presented by Blelloch~\cite{blelloch,blelloch2}, who has shown that many computational problems can be expressed by prefix aggregation. Our result does not add just one more such problem, but, by transitivity, all problems which have been previously shown to be expressible by aggregation over dominated points. An interesting theoretical conclusion can be drawn, that prefix aggregation suffices to express aggregation over dominated points, i.e., its own multidimensional generalisation. 

Seen from another perspective, our work significantly simplifies our earlier algorithm presented in Sroka et~al.~\cite{sroka2017towards}. It is a constant-rounds algorithms for solving the counting variant of the problem, written for MapReduce, and esigned to be \emph{minimal} in the sense proposed by Tao et al.~\cite{Tao2013}, guaranteeing even distribution of computation among worker nodes. It distributes range trees explicitly, and we use the same method here. However, it uses recursion to deal with consecutive dimensions and minimal group-by method from~\cite{Tao2013} to aggregate the counts. We regard it as a new contribution that all data processing tasks of this algorithm are replaced by invocations of a very few simple, high level primitives of well-understood behaviours, and that this result neatly connects two primitives, each one with its own history of research of algorithms it can express. 

This switch from a particular parallel model to high-level parallel primitives makes the algorithm simpler to understand, reducing the number and level of details which must be taken care of. Another benefit is the fact that the algorithm now avoids any direct references to the mechanisms of the parallel hardware it is running on, like processors, messages, shared resources, etc. It requires exactly those, which are used by the underlying implementations of of the primitives we rely on. 
Among them, prefix aggregation is the only one, which allows combining values from an unbounded number of data elements together. However, the role of flat maps is also crucial, because they distribute certain computations, the results of which prefix aggregation later reduces. Finally, scans allow distributing many algorithms which in the centralised setting are based on sorting and iterating over data. Such approach has many advantages similar to those pointed out in~\cite{Tao2013}, e.g., the resulting algorithms have strong guarantees concerning the way they distribute work and load. The ideas we present this paper can be viewed as generalisation of such approach to multidimensional setting the same way as range tree generalises binary search tree.

There are two other notable parallel algorithms for aggregation over dominated points. The one from Hu et al.~\cite{hu2017}[Theorem 5] and is written for MPC model and achieves load $O(Np^{-1}\log^{m-1}p)$ with $p$ processors. 
There is also an algorithm by Yufei Tao~\cite{Tao2018}[Theorem 5] that uses an entirely different idea. It is directly tailored for the MPC model and takes care of reducing the maximal amount of communication between processors. It is based on partitioning space into fragments, recursively solving problem over them and finally aggregating partial results. It achieves the optimal load $m^{O(m)}N/p$ with $p$ processors. As far as we can understand, neither of them cannot be adapted to use prefix aggregation as its main mechanism. 

\section{Primitives}

\subsection{Data model}

We assume the data to be stored in immutable but ordered lists. We are going to transform lists into new lists. Such approach and immutability is typical for distributed architectures, while ordering can be achieved by imposing some order on nodes in the cluster and distributing values such that successive nodes have increasing elements. Initial ordering of the input data is arbitrary, but we assume that input values are equipped with some numerical IDs that define it and can be used to break ties if needed.

The initial list of data points (vectors in $\RR^m$) is going to be referred to as $D$ and the list of queries as $Q$. The weights of data points are represented by weight function $w:D\to A$. To make the exposition simpler, we assume that weights are defined for queries, too, and $w(q)=e$ for $q\in Q$, so that they do not interfere with aggregation.

\subsection{Primitives and macros}

In this section we postulate primitive operations that are used to express our algorithms as well we define some convenient macros that combine them.

\begin{definition}[Sort]
    For $x=[x_0,\ldots,x_n]$ and some linear order relation $\preceq\subseteq X\times X$
\[\sort(x, \preceq) = [x_{i_1}, \ldots, x_{i_n}],\] where the multisets $\{\!\!\{x_{i_1}, \ldots x_{i_n}\}\!\!\}$ and $\{\!\!\{x_1, \ldots x_n\}\!\!\}$ are equal, and $x_{i_j}\preceq x_{i_{j+1}}$ for all j. 
\qed
\end{definition}

It is known that radix sort and quicksort are expressible by prefix aggregation~\cite{blelloch,blelloch2}, hence we could theoretically eliminate sorting from the list of primitives we rely on.

\begin{definition}[FlatMap]
    For $x=[x_0,\ldots,x_n]$ and $f:X\to [Y]$, which applied to an element produces a list of elements as the result:
\[\flatmap(x, f) = f(x_0)+\ldots+f(x_n),\] where $+$ is list concatenation.
\qed
\end{definition}

For convenience we define map, which expects $f:X \to Y$ to produce single elements, by using FlatMap and composing $f$ with list constructor $list()$: \[\map(x, f) := \flatmap(x, list() \circ f),\]
so that 
\[\map([x_1,\dots,x_n],f)=[f(x_1),\ldots,f(x_n)]\]

Zip is an operation which takes two (or more) lists of equal length and combines them into a single list of tuples, created from elements at the same positions.

\begin{definition}[Zip]
For lists $x^1=[x^1_1, \ldots, x^1_n],\ldots,x^k=[x^k_1, \ldots, x^k_n]$:
\[\zip(x^1,\ldots,x^k)=[(x^1_1,\ldots,x^k_1),\ldots, (x^1_n,\ldots,x^k_n)].\]
\qed
\end{definition}

As usually immediately after zip we want to do something with those tuples, we define macros, which map or flatmap a provided function $f:X^k\to Y$ or $f:X^k\to [Y]$ on the tuples immediately: 
\[\mzip(x^1,\ldots,x^k,f):=\map(\zip(x^1,\ldots,x^k), f),\]
\[\fmzip(x^1,\ldots,x^k,f):=\flatmap(\zip(x^1,\ldots,x^k), f).\]

We define three variants of prefix aggregation of a list $[a_1,\ldots, a_n]$ of elements of $A$, known from literature.

\begin{definition}[Scan]
For an aggregation operation $\F:A\times A\to A$, its natural element $e$ and a list $[a_1,\ldots, a_n]$ of elements of $A$:
\[\s([a_1,\ldots, a_n], \F)=[\F(a_1),\F(a_1,a_2),\ldots,\F(a_1,\ldots,a_n)],\]
\[\s^-([a_1,\ldots, a_n], \F)=[\F(e),\F(e,a_1),\ldots,\F(e,a_1,\ldots,a_{n-1})].\]\qed
\end{definition}

As we need $\s^-$ only for aggregating nondecreasing lists of numerical values with $\F=\max$ and $e=-\infty$ we define a handy macro: $\shift([x_1,\ldots,x_n]) := \s^-([x_1,\ldots,x_n],{\max})=[-\infty,x_1,\ldots,x_{n-1}]$, which is indeed a right-shift of its input if $[x_1,\ldots,x_n]$ is nondecreasing.

Another useful macro for lists of numerical values with $\F=\max$ and $e=-\infty$ is the following: $\bm([x_1,\ldots,x_n]) := \s(\sort([x_1,\ldots,x_n],\geq),{\max})=[x,\ldots,x]$ where $x=\max(x_1,\ldots,x_n)$. This macro indeed broadcasts the maximal value in a list to all positions. By zipping this list with another list we assure that the maximal value can be used for local processing of the latter, by $\map$.

\begin{definition}[Segmented scan]
For an aggregation operation $\F:A\times A\to A$ and two lists $[a_1,\ldots, a_n]$ of elements of $A$ and $[t_1,\ldots, t_n]$ of elements of some other set $T$, where the latter list is sorted:
\[\ss([a_1,\ldots, a_n],[t_1,\ldots, t_n],\F) = [\F\nolimits_{\stackrel{i\leq 1}{t_{i}=t_{1}}} a_{i}, \F\nolimits_{\stackrel{i\leq 2}{t_{i}=t_{2}}} a_{i}, \ldots, 
\F\nolimits_{\stackrel{i\leq n}{t_{i}=t_{n}}} a_{i}].\]
\end{definition}

This means, that we essentially decompose the first argument list into maximal segments over which the corresponding elements of the second list remain identical, and then compute $\s(s,{\F})$ for each segment $s$ separately, e.g., $\ss([1,2,3,4,5,6],[0,0,1,1,1,2],\Sigma)=[1,3,3,7,12,6]$.

Segmented scan can be expressed by standard scan (see~\cite{blelloch}), but is typically designed and implemented independently, which gives a chance for better performance, in particular on complex, constrained architectures, such as GPU~\cite{segmented}. It can also be implemented as a data oblivious algorithm, whose memory access pattern is independent of the actual data being processed~\cite{oblivious}. Note that scan is our only operation for combining unbounded number of elements, here by aggregation into a single value.

The primitives described in this section essentially define our model of hardware the algorithm is running on.

\section{Checking dominance by polylogarithmic data expansion}

In this section we present an important tool we need in our algorithm. It allows us to distribute the process of checking dominance relations later on. It can be viewed as a method to distribute a range tree that allows to query multidimensional data. It derives from the paper~\cite{sroka2017towards}.
 
We consider natural numbers written in binary notation, padded to some fixed length with leading 0's. Let $x$ be a bitstring. Then let $P0(x)$ be the set of all bitstrings $v$ such that $v0$ is a prefix of $x$, including the empty prefix, should $x$ start with a 0, e.g., $P0(01010)=\{0101, 01, \varepsilon\}$. Similarly, let $P1(x)$ be the set of all bitstrings $v$, such that $v1$ is a prefix of $x$.

\begin{lemma}\label{L1}
Suppose $x$ and $y$ are natural numbers represented as bitstrings of equal length, perhaps with leading 0's.

If $x<y$ then $P0(x)\cap P1(y)$ has exactly one element, and if $x\geq y$ then $P0(x)\cap P1(y)=\emptyset$. 
\end{lemma}
\begin{proof}
$x<y$ iff their longest common prefix is followed by 0 in $x$ and by 1 in $y$. Hence $P0(x)\cap P1(y)$ is nonempty iff $x<y$, which takes care of the $x\geq y$ part.

To rule out the possibility that $x<y$ and $P0(x)\cap P1(y)$ has more than 1 element, it is enough to observe that the longest element in $P0(x)\cap P1(y)$ is at the same time the shortest one because it has to be followed by different symbols in $x$ and $y$.
\end{proof}

Let $(x_1,\ldots,x_n)$ be a tuple of bitstrings. Define $P0((x_1,\ldots,x_n))=P0(x_1)\times\ldots\times P0(x_n),$ and similarly $P1((x_1,\ldots,x_n))=P1(x_1)\times\ldots\times P1(x_n).$

\begin{lemma}\label{L2}
Let $\vec{x}=(x_1,\ldots,x_n)$ and $\vec{y}=(y_1,\ldots,y_n)$ be two tuples of natural numbers encoded as bitstrings, coordinate-wise of equal lengths.

If $\vec{x}$ is coordinate-wise smaller than $\vec{y}$ then $P0(\vec{x})\cap P1(\vec{y})$ is a singleton, and otherwise it is empty. 
\end{lemma}
\begin{proof}
Follows from Lemma~\ref{L1}.
\end{proof}

\section{Algorithm}

We present the algorithm in several groups of numbered instructions and for each add explanations.

We are going to be working on data and query points together, so first the union $DQ=D\cup Q$ of those two lists is created. 
\begin{mdframed}
\begin{tabbing}
0~~~~\=\ $DQ=\fmzip(D,Q,\lambda x,y.[x,y])$\\
\end{tabbing}
\end{mdframed}

The next group of instructions is used to compute, for each  dimension, the rank of each tuple's coordinate in that dimension and the total number of unique coordinates.

Let $less:\RR\times \RR\to \RR$ be defined as 
\[less(a,b)=\left\{\begin{array}{ll}1&\text{if}~a<b\\0&\text{if}~a\geq b\end{array}\right..\]

For each dimension $i=1,\ldots,m$ we add the following instructions.

\begin{mdframed}
\begin{tabbing}
$5i-4$~~~~\=    $S_i=\map(DQ,\lambda x.x[i])$\\
$5i-3$\>    $T_i=\sort(S_i,\leq)$\\
$5i-2$\>    $W_i=\shift(T_i)$\\
$5i-1$\>    $R_i=\s(\mzip(W_i,T_i,less),\Sigma)$\\
$5i$\>    $U_i=\bm(R_i)$
\end{tabbing}
\end{mdframed}

% Line 1 is not a real loop, because dimension $m$ is constant and fixed. In particular, $U_i$ and $R_i$ do not indicate that we are using an array of lists, but are plain identifiers. 

Line $5i-4$ extracts the sequence of $i$-th coordinates of vectors from $DQ$, which is then sorted in line $5i-3$, so that it can be shifted by one position to the right in line $5i-2$. Then $less$ in line $5i-1$ essentially compares each element of $T_i$ with its predecessor and produces a list of $1s$ and $0s$, with $1$ on positions with a difference and $0$ otherwise. Therefore prefix sum of that sequence computes ranks of the elements of $T_i$. In particular, on the last index there will be total number of unique elements. We need a list with this value present at every position. It is computed in line $5i$ with $\bm$.

Now we transform each rank into its binary representation of fixed length which can be viewed as coordinates of that value in a binary search tree.

Let $bin(n,k)$ for $n\leq k$ be defined as a binary expansion of $n$ using exactly $\lceil \log k\rceil$ binary digits, i.e., with leading zeros if necessary. 

Again for each dimension $i=1,\ldots,m$ we add the following instructions.
\begin{mdframed}
\begin{tabbing}
% 7~~~~\=\ for \= each dimension $i=1,\ldots,m$:\\
$5m+i$~~~~    $BR_i=\mzip(R_i,U_i,bin)$\\
\end{tabbing}
\end{mdframed}

The instructions in lines $5m+1,\ldots,6m$ transform the original data in each dimension to the rank space, where values are replaced by their ranks within the whole dataset, written in binary. Ranks isomorphically preserve all inequalities of the original real values, and hence preserve the results of aggregations to be computed, too. 

Let $P0(x)$ and $P1(x)$ be as in Lemma~\ref{L2}.
Furthermore, let $P$ be a function from $m$-tuples of binary strings to sets of $m$-tuples of binary strings, defined as follows:
\[P(b_1,\ldots,b_m)=
\left\{\begin{array}{ll}
P0(b_1,\ldots,b_m)&\text{if $(b_1,\ldots,b_m)$ is a data point;}\\
P1(b_1,\ldots,b_m)&\text{if $(b_1,\ldots,b_m)$ is a query point.}
\end{array}
\right.\]

\begin{mdframed}
\begin{tabbing}
$6m+1$~~~~ $EDQ=\fmzip(BR_1,\ldots,BR_m,P)$
\end{tabbing}
\end{mdframed}

This results in a new, expanded sequence $EDQ$ consisting of tuples of bitstrings. We assume they inherit ID and weights from their originating data and query points. This sequence will contain many duplicated tuples, but with different IDs. This operation increases the total size of data from $O(n)$ to $O(n\log^mn).$

Let $\leq_{lex}$ be doubly lexicographic ordering relation on tuples from $EDQ$: lexicographic for bitstrings in each coordinate and lexicographic between  coordinates, with the (crucial) additional requirement, that in case of a tie of several tuples data points precede queries.
\begin{mdframed}
\begin{tabbing}
$6m+2$~~~~     $SEDQ=\sort(EDQ,\leq_{lex})$\\
\end{tabbing}
\end{mdframed}

$SEDQ$ is composed of segments consisting of duplicates, differing only by IDs and weights. Data points come before queries among each segment of equal values.  At this point we have prepared all data we need, it remains to perform several steps of aggregations.

\begin{mdframed}
\begin{tabbing}
$6m+3$~~~~ $A_1=\ss(\map(SEDQ,\lambda x.w(x)),SEDQ,{\F})$\\
\end{tabbing}
\end{mdframed}

In line $6m+3$ we compute segmented prefix sums of the weights of the sorted data and queries. All bitstring coordinates serve as tags for segmentation. This way a list is created, where each query point $q$ corresponds (by position on the list) to the aggregation of all weights of data points which yielded the same tuple of bitstrings. By Lemma~\ref{L2}, each such identical tuple originates form a data point $d$ such that $d<q,$ and, moreover, for each fixed tuple $t$ derived from $q$ there is one-to-one correspondence between such tuples and data points which also produced $t$ and are (therefore) dominated by $q$.

Queries have been assigned neutral weight, so they do not interfere with the scan. However, some aggregation happens at the positions of data points, too.

Let $\leq_{ID}$ be ordering relation on tuples which compares their associated $ID$ values. 

\begin{mdframed}
\begin{tabbing}
$6m+4$~~~~\= $A_2=\sort(A_1,\geq_{ID})$\\
$6m+5$\>     $A_3=\ss(A_2,\map(SEDQ,\lambda x.ID(x)),\F)$
\end{tabbing}
\end{mdframed}

Line $6m+4$ is a sort of partial aggregations by $ID$ in nonincreasing sequence, which creates a continuous segment of aggregation values corresponding to each query. There are also separate segments of data points, which are irrelevant now. 

After that in line $6m+5$ we compute segmented prefix sums of the already partially aggregated weights of the sorted data and queries, whereby their $ID$ values serve as tags for segmentation. This way partial aggregations for each query are further aggregated, so that the total aggregation of each query $q$ corresponds, by position, to the last tuple which originated from $q.$ This total aggregation for $q$ comes from all $d$ such that $P(q)\times P(d)\neq \emptyset$ (those sets are created in line $6m+1$). By Lemma~\ref{L2}, $\{d|P(q)\cap P(d) \neq \emptyset\}=\{d|d<q\}$, and each such $d$ in the r.h.s. is witnessed by exactly one element of $P(q)\cap P(d)$. 

At this point the aggregation of each query is already computed, but the data contains many partial aggregations for the same query, too. Therefore the last task is to distribute the total aggregation of each query to all tuples with the same ID. This task can be significantly simplified if $A$ is linearly ordered and $F(a,b)\geq a$ for all $a,b\in A.$ In this case using a segmented variant of $\bm$ with ID defining segmentation does immediately the work. Otherwise the method to achieve the goal is more complex and described below. 

Let $neutral\_if\_eq(id_1,id_2,a)=\left\{\begin{array}{ll}e&\text{if}~id_1=id_2\\a&\text{if}~id_1\neq id_2\end{array}\right..$

\begin{mdframed}
\begin{tabbing}
$6m+6$~~~~\= $A_4=\sort(A_3,\leq_{ID})$\\
$6m+7$\>     $H=\shift(\map(A_4,\lambda x.ID(x)))$\\
$6m+8$\>     $A_5=\mzip(A_4,H,\map(A_4,\lambda x.ID(x)),neutral\_if\_eq)$\\
$6m+9$\>     $Aggregated=\ss(A_5,\map(\lambda x.ID(x),A_5),\F)$
\end{tabbing}
\end{mdframed}

In line $6m+6$ we revert the sort order of $A_4.$ Now the complete aggregations come at the beginning of each segment, and, moreover, the $ID$ values are nondecreasing, hence we can shift them in line $6m+7$. Line $6m+8$ resets all computed weights to the neutral $e$, except at the beginning of each segment, where the complete aggregation is present. Finally, segmented scan in line $6m+9$ aggregates these values with neutral elements elsewhere. 

Now each query ID is accompanied by the aggregation of weights of its dominated points, which means that the desired output has been computed. In total it took $6m+9$ instructions and at most that many intermediate lists of data created. 

\section{Improvement by one logarithm}

It is possible to reduce the amount of data generated by a factor of $\log n$. 

One of the coordinates (say: the last one) is chosen. It is not replaced by ranks and left in the form of real numbers. The remaining ones are replaced by ranks and prefixes are generated, exactly as in the basic algorithm, from both data points and queries. This results in a multiset of $\leq n\log^{m-1} n$ tuples with $m-1$ coordinates in the form of bitstrings and the last, $m$-th coordinate being real number. Identifiers are retained. 

Now sort the data and queries $EDQ$ into $SEDQ$ (see line $6m+2$ and explanation thereof above) according to the doubly lexicographic order, with queries preceding data points in case of equality of all coordinates (including the $m$-th).

This results in segments of data elements with $m-1$ first coordinates equal, sorted according to the last coordinate within the segment. Then the remainder of the algorithm is executed exactly as in the basic version.

\section{Relation to range trees and complexity}

The algorithm we have presented above is derived from our map-reduce algorithm~\cite{sroka2017towards}, and is indeed a parallelisation of the common sequential algorithm, based on range trees. The move to the rank space with ranks expressed as binary expansions is equivalent to speaking about elements in terms of their positions in a balanced binary tree, whose leaves hold the sorted data. The binary encodings then correspond to branches in the tree, and their prefixes to the positions where attached trees of smaller dimension are located. 

Our algorithm inherits its total data complexity of $O(n\log^{m-1}n)$ from the range tree algorithm. This distributed data structure is generated by flat maps, while the parallelisation is achieved by expressing operations on the tree in terms of sorting, zipping and prefix aggregation. In the older paper we didn't present complexity analysis, only tests of an implementation on MapReduce, for which it was designed. In particular, the running times did scale linearly with the number of machines. 

Time complexity of the present algorithm depends very much on the underlying architecture and complexity of the primitive operations, but its analysis is pretty straightforward in each case, since it is a fixed length sequence of operations of very well known properties, applied to lists of data of sizes easy to determine. In particular, no matter what architecture it is executed on, if the implementations of primitives do the same total work as their sequential variants, then the whole algorithm will also have the total work of the sequential algorithm using range trees.
Indeed, in the sequential case the only nonlinear (and thus dominating) operation is sorting, and the size of the data is $O(n\log^{m-1}n)$. One linearithmic sort takes time $O(n\log^{m-1}n\cdot \log(n\log^{m-1}n))=O(n\log^{m}n)$,
which is equal to the worst case of the standard sequential implementation, calculated as creating the range tree with $n/2$ data points and then processing $n/2$ queries by this tree. There is a one logarithm better variant of Chazelle~\cite{chazelle1990}, which however works only for counting.

For the MPC model, it is known that for $\delta>0$, sorting and scanning of $n$ values can be performed deterministically in a constant number of rounds using $n^\delta$ space per machine, $O(n)$ total space, and $poly(n)$ local computation, which follows directly from analogous bounds for MapReduce computation. The load can be made $O(N/p)$~\cite{goodrich}. This implies that our algorithm in dimension $m$ can be implemented deterministically in $O(1)$ rounds, with $n^\delta$ space per machine, $O(n\log^{m-1}n)$ total space and $poly(n)$ local computation, with load $O(N\log^{m-1}N/p).$ 

By comparison, the algorithm by Hu et al.~\cite{hu2017} achieves load $O(Np^{-1}\log^{m-1}p)$ with $p$ processors, which is better than what get in this paper. This is not surprising, since our reliance on high-level primitives gives us much less freedom in designing the computation mechanism and achieving low loads. In particular, the set of instructions we use does not permit shifting the computational effort into the local computation, which is the method to lower the loads. The algorithm of Tao~\cite{Tao2018} does it to an even higher extent, arriving at load $m^{O(m)}N/p.$

\section{Summary and future research}

In this paper we have presented algorithm for aggregating over dominated points in $\RR^m$, where $m$ is constant. Our algorithm is based on a limited set of primitive operations: sorting, prefix aggregations, zip and flat maps. All those primitive operations are well studied and their efficient implementations exist for essentially all distributed architectures. 

This proves that one-dimensional prefix aggregation allows expressing its own multidimensional generalisation. 
The latter problem has many practical applications, as well as it is known to be a parallel primitive, allowing to express in turn further problems. By transitivity, our result expresses all those problems in terms of the above mentioned primitive operations.

We consider our result to be primarily of theoretical interest. However, our algorithm might actually lead to practical implementations. First of all, it is absolutely transparent and does not hide any significant computation steps. The local computation is on the level of individual tuples, only. No large collections of data need to be broadcasted to computation nodes and there is no limit on number of such nodes. Otherwise we use high-level primitives of very well understood algorithmic properties, and whose optimised implementations exist for virtually all hardware platforms. Also the organisation of the algorithm into a sequence of functional operations (without any branch) on immutable ordered lists is very convenient for implementation. Last but not least, the choice of primitives guarantees that the algorithm has several desired properties, e.g., it is minimal in the sense of~\cite{Tao2013}.

An obvious item on the ``further research'' list is undertaking experiments with the algorithm on diverse parallel platforms. 

%-Może dalej co planujemy: analiza złożoności przy założeniu jakiejś założoności operacji podstawowych; porównanie dla średnich m, gdzie inne algorytmy mają problem z rozmiarem hiperkostki (może to w ogóle dopiszemy gdzieś wcześniej)?
%-Może jakoś nam się uda skomentować czemu ma sens wyrażać złożoność przy pomocy liczby wywołań operacji podstawowych (przy rozpraszaniu sort, scan i zip mogą powodować wymianę danych między komputerami przy czym scan bardzo ograniczoną, a flat map jest lokalny i można go nie liczyć). Jak tego teraz nie ogarniemy to można napisać to jako todo.

\end{document}